\documentclass[conference]{IEEEtran}
\usepackage{subfigure}
\usepackage{amsmath}
\usepackage{graphicx,amssymb,lineno,bm,subfigure}
\usepackage{algorithm}
\usepackage{algorithmic}

\newtheorem{rmk}{\textbf{Remark}}

\newtheorem{prop}{\textbf{Proposition}}
\newtheorem{proof}{\textbf{Proof}}
\newtheorem{lma}{\textbf{Lemma}}
\IEEEoverridecommandlockouts
\begin{document}

%
\title{\huge{ARQ with Adaptive Feedback for Energy Harvesting Receivers}}
\author{\IEEEauthorblockN{Yuyi Mao{$^\dagger$}, Jun Zhang{$^\dagger$}, and K. B. Letaief{$^{\dagger\ast}$}, \emph{Fellow, IEEE}}
\IEEEauthorblockA{$^\dagger$Dept. of ECE, The Hong Kong University of Science and Technology, Clear Water Bay, Hong Kong\\
$^{\ast}$Hamad bin Khalifa University, Doha, Qatar\\
Email: \{ymaoac, eejzhang, eekhaled\}@ust.hk}
\thanks{This work is supported by the Hong Kong Research Grants Council under Grant No. 610212.}
}
\maketitle

\begin{abstract}
Automatic repeat request (ARQ) is widely used in modern communication systems to improve transmission reliability. In conventional ARQ protocols developed for systems with energy-unconstrained receivers, an acknowledgement/negative-acknowledgement (ACK/NACK) message is fed back when decoding succeeds/fails. Such kind of non-adaptive feedback consumes significant amount of energy, and thus will limit the performance of systems with energy harvesting (EH) receivers. In order to overcome this limitation and to utilize the harvested energy more efficiently, we propose a novel ARQ protocol for EH receivers, where the ACK feedback can be adapted based upon the receiver's EH state. Two conventional ARQ protocols are also considered. By adopting the packet drop probability (PDP) as the performance metric, we formulate the throughput constrained PDP minimization problem for a communication link with a non-EH transmitter and an EH receiver. Optimal reception policies including the sampling, decoding and feedback strategies, are developed for different ARQ protocols. Simulation results will show that the proposed ARQ protocol not only outperforms the conventional ARQs in terms of PDP, but can also achieve a higher throughput.
\end{abstract}

%
\IEEEpeerreviewmaketitle
\section{Introduction}
Energy harvesting (EH) has recently emerged as a promising technique to provide renewable energy sources for wireless systems \cite{YMao1506}. As EH devices can capture ambient recyclable energy, such as solar and wind energy, the lifespan of communication systems can be prolonged significantly. However, due to their intermittent and sporadic nature, conventional communication protocols may not be capable of securing full benefits of the harvested energy \cite{SUlukus1503}. Consequently, protocols tailored for EH communication systems have been developed in recent years \cite{OOzel1109}-\cite{Aprem1310}. Nevertheless, most of the existing works focus on systems with EH transmitters, and few results are available for EH receivers.

In some scenarios, e.g. communicating over short distances, high data rates can be achieved with a relatively small transmit power, and the energy consumption associated with the complex detection and decoding operations at the receiver becomes dominant \cite{SCui0509}. This motivates the recent investigation of communication systems with EH receivers. The optimal packet sampling and decoding policies for an EH receiver to maximize packet throughput were developed in \cite{Yates1503}. By jointly considering the EH transmitter and receiver, i.e., a dual EH link, a general utility maximization problem was solved in \cite{Tutuncuoglu12}. In order to improve transmission reliability, automatic repeat request (ARQ) based packet retransmission has been introduced for systems with EH receivers \cite{MSharma1412}-\cite{AYadav1506}. In \cite{MSharma1412}, a dual EH link was modeled as a discrete-time Markov chain, and the packet drop probabilities (PDPs) under basic ARQ and hybrid ARQ with chase combining (HARQ-CC) were analyzed, respectively. A suboptimal power allocation policy was proposed for dual EH links to minimize the PDP under an acknowledgement (ACK) based ARQ in \cite{SZhou1503}. Moreover, to avoid wasting of the harvested energy, a selective-sampling scheme that allows the EH receiver to sample part of a packet was proposed in \cite{AYadav1506}.

So far, the ARQ protocols adopted  for systems with EH receivers, e.g. \cite{MSharma1412}-\cite{AYadav1506}, are inherited from the conventional ones developed for energy-unconstrained receivers \cite{PWu1112}, where an ACK message should be fed back when decoding succeeds, and a negative-acknowledgement (NACK) message should be sent back to the transmitter when decoding fails. Such kind of aggressive and non-adaptive feedback strategy consumes significant amount of energy, which may lead to energy outage for sampling and decoding, and thus will limit the performance of systems with energy-scarce EH receivers.

To overcome this limitation, a novel ARQ protocol with adaptive feedback will be proposed in this paper, which enables ACK feedback management and helps utilize the harvested energy more efficiently. As a comparison, two conventional ARQ protocols will also be studied, one without feedback and one with non-adaptive feedback. By adopting the PDP as the performance metric, we formulate the throughput constrained PDP minimization problem for a point-to-point wireless link with a non-EH transmitter and an EH receiver. The optimal reception policies, including the sampling, decoding and feedback strategies, are investigated. For the ARQ protocol without feedback, a myopic policy is shown to be optimal. For ARQs with feedback, an optimal iterative algorithm is developed with the aid of the \emph{Dinkelbach approach} and \emph{Markov decision process} (MDP) techniques. Simulation results demonstrate the effectiveness of the obtained optimal reception policies. It is shown that the proposed ARQ protocol not only outperforms the conventional ARQs in terms of PDP, but it is also competent to achieve a higher throughput.

The organization of this paper is as follows. We introduce the system model and the proposed ARQ protocol in Section II, while the throughput constrained PDP minimization problem is formulated in Section III and the optimal reception policies for different ARQ protocols are developed in Section IV. Simulation results are presented in Section V, and we will conclude this paper in
Section VI.

\section{System Model}

We consider a communication link with a non-EH transmitter and an EH receiver. Time is slotted, and the time slot length, system bandwidth and background noise variance are normalized. We assume the channel experiences independent and identically distributed (i.i.d.) block fading among different time slots, and denote the channel power gain at the $t$th time slot as $|h_{t}|^{2}$, with the cumulative distribution function (CDF) given by $F_{H}\left(x\right)$. 

\subsection{Energy Model at the Receiver}
The EH process is modeled as successive energy packet arrivals, i.e., at the beginning of the $t$th time slot, $E_{H,t}$ units of energy arrives at the EH receiver. Without loss of generality, we assume $\{E_{H,t}\}$ are i.i.d. discrete random variables, which take values from set $\mathcal{E}\triangleq \{e_{n}\}$ ($e_{n}$'s are the possible amounts of the harvested energy), and distribute according to $\mathbb{P}_{\mathcal{E}}\left(e\right),e\in\mathcal{E}$. A battery with capacity $B_{\max}$ is deployed to store the harvested energy, and the battery energy level $b_{t}$ at the beginning of the $t$th time slot evolves according to
\begin{equation}
b_{t+1}=\min\{b_{t}-E_{c,t}+E_{H,t+1},B_{\max}\},t=1,2,\cdots
\label{batterydyn}
\end{equation}
where $b_{1}=E_{H,1}$, and $E_{c,t}$ denotes the energy consumed by the receiver in the $t$th time slot.

The energy spent on sampling and decoding are comparable, and both should be included in the receiver's energy consumption model \cite{Yates1503}. We denote the energy consumption for each sampling and decoding operation as $E_{\rm{s}}$ and $E_{\rm{d}}$, respectively. The energy consumed for a reliable ACK feedback, denoted by $E_{\rm{f}}$, is also taken into account. The delay incurred by decoding and feedback are ignored for simplicity, and energies are assumed to be integer multiples of a basic energy quantum, $E$, without loss of generality.

\subsection{Transmission Protocol}
The transmitter transmits packets with a constant rate $R$ and a fixed transmit power $p_{\rm{tx}}$. A maximum number of transmission attempts for each packet, denoted as $K$ ($K>1$), is set \cite{Aprem1310}, \cite{MSharma1412}-\cite{PWu1112}. In other words, the transmitter keeps transmitting the same packet either until all transmission attempts are used or the ACK signal is fed back. We denote $k\in\mathcal{K}\triangleq \{0,1,\cdots,K-1\}$ as the transmission index if a packet is in its $k$th retransmission, i.e., the $\left(k+1\right)$th transmission. When the maximum transmission attempt is used or the ACK is fed back, the transmitter will start to transmit a new packet in the next time slot, which indicates that a packet will be dropped without being received successfully within $K$ transmission attempts. In particular, the packet drop probability (PDP) is adopted as the performance metric in this paper.

\subsection{Reception Policy}
At each time slot, the receiver decides whether to sample the received signal, i.e., $a_{s,t}\in \{0,1\}$, where $a_{s,t}=1$ means it will perform sampling in the $t$th time slot, and vice versa. If a packet is sampled, both the pilot symbols for channel training and the information bits are obtained. We assume perfect channel estimation and thus the value of $|h_{t}|^{2}$ is available at the receiver. The receiver decodes based on the sample of the packet obtained in the current time slot, similar as in the basic ARQ protocol \cite{MSharma1412,SZhou1503,PWu1112}, and it knows if a packet can be successfully decoded based on the knowledge of $|h_{t}|^{2}$. Specifically, if $|h_{t}|^{2}\geq \left(2^{R}-1\right)p_{\rm{tx}}^{-1}\triangleq |h_{\rm{th}}|^{2}$, the receiver will correctly decode the packet by consuming $E_{\rm{d}}$ units of energy. Otherwise, the sampled packet will be discarded and no energy will be used. When there is ACK feedback, the receiver should further determine the feedback strategy $\{a_{f,t}\}$, where $a_{f,t}=1$ means an ACK message will be fed back in time slot $t$ for a successfully decoded packet, and $a_{f,t}=0$ indicates nothing will be fed back. We refer the sampling and feedback strategies as the \emph{reception policy}, which is denoted as $\{\mathbf{a}_{t}\}\triangleq \{\langle a_{s,t},a_{f,t}\rangle \}$, and the decoding policy is included implicitly.

\subsection{ARQ Protocols}
Before proposing the ARQ protocol with adaptive feedback, we first introduce two baseline ARQ protocols, including ARQ without feedback and ARQ with non-adaptive feedback:

\textbf{ARQ without feedback:} In this protocol, ACK feedback is not allowed, i.e., $a_{f,t}=0$, and only $a_{s,t}$ needs to be designed ($\mathbf{a}_{t}\in\{\langle 0,0\rangle,\langle 1,0 \rangle\}$). Thus, the transmitter starts to transmit a new packet every $K$ time slots, regardless of the receiver's operation. This protocol can save the feedback energy, but may reduce the system throughput.

\textbf{ARQ with non-adaptive feedback:} In this protocol, ACK feedback is mandatory in the time slots where packets are successfully decoded, i.e., $\mathbf{a}_{t}\in\{\langle 0,0\rangle,\langle 1,1 \rangle\}$. Such an ACK feedback strategy is commonly-used in conventional ARQ protocols, e.g. \cite{Aprem1310}, \cite{MSharma1412}-\cite{PWu1112} and references therein.

\textbf{ARQ with adaptive feedback:} In the proposed ARQ protocol with adaptive feedback, the ACKs can be delayed or eliminated for the successfully decoded packets, which is more flexible compared to ARQ with non-adaptive feedback. Hence, $\mathbf{a}_{t}\in\{\langle 0,0\rangle, \langle 0,1 \rangle, \langle 1,0 \rangle,\langle 1,1 \rangle\}$, where $\mathbf{a}_{t}=\langle 0,1\rangle$ means the ACK is delayed for a packet decoded successfully in a previous time slot and is fed back in the $t$th time slot. Actually, the ARQs without feedback and with non-adaptive feedback can be regarded as special cases of the proposed ARQ protocol.

\begin{rmk}
With conventional energy-unconstrained receivers, delaying or eliminating the ACKs brings no benefit, i.e., ARQ with adaptive feedback reduces to ARQ with non-adaptive feedback. However, with EH receivers, the proposed ARQ protocol offers the option of feedback management, and thus helps improve the efficiency of utilizing the harvested energy. As will be seen in the coming sections, with the optimal reception policies, the proposed ARQ protocol improves both the PDP and throughput performance compared to the two baseline ARQ protocols.
\end{rmk}

\section{Problem Formulation}
Due to the intermittent and sporadic nature of EH, the reception policies should be carefully designed in order to maximize the system performance. In this section, a throughput constrained PDP minimization problem will be formulated, and its optimal solution will be developed in Section IV.

The PDP is defined as the ratio of the number of dropped packets to the number of transmitted data packets \cite{SZhou1503} as
\begin{equation}
p_{\rm{drop}}\triangleq \lim_{T\rightarrow \infty}\frac{\mathbb{E}\left[\sum_{t=1}^{T}\bm{1}_{D_{t}}\right]}{\mathbb{E}\left[\sum_{t=1}^{T}\bm{1}_{N_{t}}\right]}\overset{\left(\dag\right)}{=}
\frac{\lim\limits_{T\rightarrow \infty}\frac{1}{T}\mathbb{E}\left[\sum_{t=1}^{T}\bm{1}_{D_{t}}\right]}{\lim\limits_{T\rightarrow \infty}\frac{1}{T}\mathbb{E}\left[\sum_{t=1}^{T}\bm{1}_{N_{t}}\right]},
\label{defoutage}
\end{equation}
where $\bm{1}_{X}=1$ if event $X$ happens, and $D_{t}$ ($N_{t}$) denotes the event of packet drop (starting a new packet) at the $t$th time slot.\footnote{$\left(\dag\right)$ holds under the condition that both the limits of the numerator and denominator exist, which are assumed in this work.} Accordingly, the throughput, i.e., the average number of successfully received packets per time slot, is defined as
\begin{equation}
\mathcal{T} \triangleq \lim_{T\rightarrow \infty}\frac{1}{T}\mathbb{E}\left[\sum_{t=1}^{T}\bm{1}_{S_{t}}\right],
\label{deftp}
\end{equation}where $S_{t}$ denotes the event of decoding a packet correctly at time slot $t$. It can be verified that the following identity holds:
\begin{small}
\begin{equation}
\lim_{T\rightarrow \infty}\frac{1}{T}\mathbb{E}\left[\sum_{t=1}^{T}\bm{1}_{N_{t}}\right]=
\lim_{T\rightarrow \infty}\frac{1}{T}\mathbb{E}\left[\sum_{t=1}^{T}\left(\bm{1}_{D_{t}}+\bm{1}_{S_{t}}\right)\right].
\label{SNDidentity}
\end{equation}
\end{small}This is because the transmitter starts to transmit a new packet either if the ACK signal is received or the $K$th transmission attempt is used.

Denote the system state at the beginning of the $t$th time slot as $\mathbf{s}_{t}$, which can be represented by a triplet, i.e., $\mathbf{s}_{t}\triangleq \langle b_{t}, k_{t}, i_{t}\rangle $, where $b_{t}\in \{0,E,\cdots,ME\}$ ($ME\triangleq B_{\max}$) is the battery energy level, $k_{t}\in\mathcal{K}$ is the transmission index, and $i_{t}\in\{0,1\}$ denotes the reception state. In particular, $i_{t} = 0$ indicates the current transmitting packet has not been decoded correctly, and vice versa. A reception policy is mathematically characterized by a mapping from the system state to the action, i.e., $\Psi:\mathbf{s}\rightarrow \mathbf{a}$. We denote $d\left(\mathbf{s}_{t},\mathbf{a}_{t}\right)$, $n\left(\mathbf{s}_{t},\mathbf{a}_{t}\right)$ and $s\left(\mathbf{s}_{t},\mathbf{a}_{t}\right)$ as the expected values of $\bm{1}_{D_{t}}$, $\bm{1}_{N_{t}}$ and $\bm{1}_{S_{t}}$ when the system is in state $\mathbf{s}_{t}$ while action $\mathbf{a}_{t}$ is taken. Thus, the throughput constrained PDP minimization problem \cite{PWu1112} is formulated as
\begin{align}
&\mathcal{P}_{1}: \min_{\Psi} \frac{\lim\limits_{T\rightarrow \infty}\frac{1}{T}\mathbb{E}\left[\sum_{t=1}^{T}d\left(\mathbf{s}_{t},\mathbf{a}_{t}\right)\right]}{\lim\limits_{T\rightarrow \infty}\frac{1}{T}\mathbb{E}\left[\sum_{t=1}^{T}n\left(\mathbf{s}_{t},\mathbf{a}_{t}\right)\right]}\label{PDPobj}\\
&\ \ \ \ \ \ \mathrm{s.t.\ }\lim_{T\rightarrow \infty}\frac{1}{T}\mathbb{E}\left[\sum_{t=1}^{T}s\left(\mathbf{s}_{t},\mathbf{a}_{t}\right)\right]\geq \mathcal{T}_{\rm{th}}\label{tpconstraint}\\
&\ \ \ \ \ \ \ \ \ \ \ \ \mathbf{a}_{t}\in\mathcal{A}_{\mathbf{s}_{t}} \label{feasconstraint},
\end{align}
where the expectations are with respect to the sample-path of $\mathbf{s}_{t}$ and $\mathbf{a}_{t}$ induced by $\Psi$. (\ref{tpconstraint}) is the throughput constraint while (\ref{feasconstraint}) stands for the feasible action set, $\mathcal{A}_{\mathbf{s}_{t}}$, when the system is in state $\mathbf{s}_{t}$. Note that $\mathcal{A}_{\mathbf{s}_{t}}$ is related to the adopted ARQ protocol, and it is further specified as $\mathcal{A}_{\mathbf{s}_{t}}^{\rm{wo}}$, $\mathcal{A}_{\mathbf{s}_{t}}^{\rm{nA}}$ and $\mathcal{A}_{\mathbf{s}_{t}}^{\rm{A}}$ for the ARQs without feedback, with non-adaptive feedback and with adaptive feedback, respectively. Generally, $\mathcal{P}_{1}$ is a stochastic optimization problem with a fractional objective function, and we will investigate its optimal solutions for different ARQ protocols in the next section.

\section{Optimal Reception Policies}
In this section, we will investigate the optimal reception policies for the three ARQ protocols. For ARQ without feedback, a myopic policy will be shown to be optimal. For ARQs with (non-adaptive or adaptive) feedback, an optimal iterative algorithm will be proposed.

\subsection{ARQ without ACK Feedback}
For ARQ without feedback, the transmitter starts to transmit a new packet every $K$ time slots, i.e., $\lim\limits_{T\rightarrow \infty}\frac{1}{T}\mathbb{E}\left[\sum_{t=1}^{T}n\left(\mathbf{s}_{t},\mathbf{a}_{t}\right)\right]=1\slash K$. Thus, we can minimize the numerator in (\ref{PDPobj}) via optimizing the sampling strategy $\{a_{s,t}\}$. Since only $a_{s,t}$ needs to be designed, we use $\mathbf{a}_{t}$ and $a_{s,t}$ interchangeably in this subsection. According to (\ref{SNDidentity}), it is equivalent to maximize the throughput. Therefore, we consider the following throughput maximization problem
\begin{align}
\mathcal{P}_{2}: \max_{\Psi} \lim\limits_{T\rightarrow \infty}\frac{1}{T}\mathbb{E}\left[\sum_{t=1}^{T}s\left(\mathbf{s}_{t},\mathbf{a}_{t}\right)\right]\ \mathrm{s.t.}\ \mathbf{a}_{t}\in\mathcal{A}_{\mathbf{s}_{t}}^{\rm{wo}},
\end{align}
and the solution for $\mathcal{P}_{1}$ can be obtained based on Lemma 1.
\begin{lma}
Denote the optimal solution and the optimal value of $\mathcal{P}_{2}$ as $\Psi^{*}_{\mathcal{P}_{2}}$ and $\mathcal{T}^{*}_{\mathcal{P}_{2}}$, respectively. If $\mathcal{T}_{\mathcal{P}_{2}}^{*}<\mathcal{T}_{\rm{th}}$, $\mathcal{P}_{1}$ is infeasible. Otherwise, $\Psi_{\mathcal{P}_{2}}^{*}$ is the optimal solution for $\mathcal{P}_{1}$.
\end{lma}
\begin{proof}
When $\mathcal{T}_{\mathcal{P}_{2}}^{*}<\mathcal{T}_{\rm{th}}$, $\mathcal{P}_{1}$ is infeasible since the maximum throughput is less than $\mathcal{T}_{\rm{th}}$. If (\ref{tpconstraint}) is added, $\mathcal{P}_{2}$ is equivalent to $\mathcal{P}_{1}$. Thus, when $\mathcal{T}_{\mathcal{P}_{2}}^{*}\geq \mathcal{T}_{\rm{th}}$, (\ref{tpconstraint}) holds under $\Psi_{\mathcal{P}_{2}}^{*}$, i.e., $\Psi_{\mathcal{P}_{2}}^{*}$ is also optimal for $\mathcal{P}_{1}$.
\end{proof}

In order to solve $\mathcal{P}_{2}$, we first specify the feasible action set $\mathcal{A}_{\mathbf{s}_{t}}^{\rm{wo}}$: If $i_{t}=0$ and $b_{t}\geq E_{\rm{s}}+E_{\rm{d}}$, we have $\mathcal{A}_{\mathbf{s}_{t}}^{\rm{wo}}=\{0,1\}$. Otherwise, we have $\mathcal{A}_{\mathbf{s}_{t}}^{\rm{wo}}=\{0\}$, which indicates that the receiver can sample only when the current transmitting packet has not been decoded successfully, and meanwhile, the available energy is sufficient for both the sampling and decoding operations\footnote{In principle, the receiver can perform sampling after a packet has been decoded correctly. However, this has no contribution to the PDP/throughput performance. It can also sample when $E_{\rm{s}}\leq b_{t} < E_{\rm{s}}+E_{\rm{d}}$. However, without the decoding capability, the sampling energy is wasted since the receiver decodes based on the sample of a packet obtained in the current time slot. Without loss of optimality, theses two actions have been precluded in $\mathcal{A}_{\mathbf{s}_{t}}^{\rm{wo}}$.}.

When $a_{s}=0$, no energy is consumed and the reception state remains unchanged (unless it is in the $K$th transmission attempt due to the start of a new packet in the next time slot), i.e., no packet is decoded in the current time slot and $s\left(\mathbf{s},a_{s}=0\right)=0$. When $a_{s}=1$, the receiver performs sampling, which consumes $E_{\rm{s}}$ units of energy. With probability $1-F_{H}\left(|h_{\rm{th}}|^{2}\right)\triangleq p_{c}$, the packet will be successfully decoded (the reception state becomes $1$ in the next time slot unless it is in the $K$th transmission) and $E_{\rm{d}}$ units of energy will be consumed by the decoder; while with probability $\overline{p}_{c}=1-p_{c}$, the channel will be in outage and the decoder will stay inactive, i.e., the reception state remains $0$. Thus, $s\left(\mathbf{s},a_{s}=1\right)=p_{c}$.

Theoretically, $\mathcal{P}_{2}$ is an MDP problem and can be solved by standard algorithms, which, however, will have high complexity. We propose a myopic policy for $\mathcal{P}_{2}$, for which the receiver performs sampling when the available energy is sufficient for sampling and decoding, while the current transmitting packet has not been received correctly, i.e.,
\begin{equation}
a_{s,t}=\bm{1}_{i_{t}=0, b_{t}\geq E_{\rm{s}}+E_{\rm{d}}},
\label{myopic}
\end{equation}
which is optimal for $\mathcal{P}_{2}$ as shown in the following proposition.
\begin{prop}
The myopic policy in (\ref{myopic}) is optimal for $\mathcal{P}_{2}$.
\end{prop}
\begin{proof}
Please refer to the Appendix.
\end{proof}

Under mild conditions, the system state constitutes an ergodic Markov chain under the myopic policy. Thus, the steady state distribution $\{\pi_{\mathbf{s}}\}$ can be obtained by solving the balance equation, and the maximum throughput is given by $\mathcal{T}_{\mathcal{P}_{2}}^{*}=\sum_{b\geq E_{\rm{s}}+E_{\rm{d}},i=0}\pi_{\mathbf{s}}p_{c}$. Together with Lemma 1, we can either identify the infeasibility or the optimal solution for $\mathcal{P}_{1}$.

\subsection{ARQ with ACK Feedback}
The reception policies design for the ARQ protocols with feedback, either non-adaptive or adaptive, is much more challenging due to the fractional objective function, as the value of $\lim\limits_{T\rightarrow\infty}\frac{1}{T}\mathbb{E}\left[\sum_{t=1}^{T}n\left(\mathbf{s}_{t},\mathbf{a}_{t}\right)\right]$ is unknown. To develop the optimal policies, we first transform $\mathcal{P}_{1}$ into $\mathcal{P}_{3}$ with a weighted linear objective function, and reveal their relationship in Lemma 2.
\begin{lma}
If $\mathcal{P}_{1}$ is feasible, the minimum PDP, $p_{\rm{drop}}^{*}$, is achieved if and only if the optimal value of the following weighted minimization problem $\mathcal{P}_{3}$ is zero when $q=p_{\rm{drop}}^{*}$:
\begin{align}
&\mathcal{P}_{3}: \min_{\Psi} \lim\limits_{T\rightarrow\infty}\frac{1}{T}\mathbb{E}\left[\sum_{t=1}^{T}\left(d\left(\mathbf{s}_{t},\mathbf{a}_{t}\right)-q \cdot n\left(\mathbf{s}_{t},
\mathbf{a}_{t}\right)\right)\right]\nonumber\\
&\ \ \ \ \ \ \mathrm{s.t.\ \ }(\ref{tpconstraint}),(\ref{feasconstraint}).\nonumber
\end{align}
Moreover, the optimal solution for $\mathcal{P}_{3}$ with parameter $p_{\rm{drop}}^{*}$, $\Psi_{\mathcal{P}_{3}}^{*}$, is also optimal for $\mathcal{P}_{1}$.
\end{lma}
\begin{proof}
The proof is similar to that for Theorem 1 in \cite{DNg1207}, which is omitted due to space limitation.
\end{proof}

As $\lim\limits_{T\rightarrow \infty}\frac{1}{T}\mathbb{E}\left[\sum_{t=1}^{T}n\left(\mathbf{s}_{t},\mathbf{a}_{t}\right)\right]\in \left[1\slash K, 1\right]$, the objective function in $\mathcal{P}_{3}$ is strictly monotonic decreasing with $q$, which suggests an iterative algorithm to obtain $p_{\rm{drop}}^{*}$ and $\Psi_{\mathcal{P}_{1}}^{*}$ (known as the Dinkelbach approach \cite{Dinkelbach6703}), as summarized in Algorithm 1. Note that if $\mathcal{P}_{3}$ can be solved optimally, Algorithm 1 is guaranteed to converge to the optimal solution, which can be proved by following a similar approach as in \cite{DNg1207,Dinkelbach6703}. Next, we will develop the optimal solution for $\mathcal{P}_{3}$ and focus on the ARQ protocol with adaptive feedback. As a special case, the solution for ARQ with non-adaptive feedback can be obtained with minor modifications.

\begin{algorithm}[h] 
\caption{Optimal Iterative Algorithm for $\mathcal{P}_{1}$ for ARQ with Adaptive Feedback.} 
\label{alg1} 
\begin{algorithmic}[1] 
\STATE Initialize the maximum number of iterations $I_{\max}$ and the maximum tolerance $\Delta>0$, set $q=1$ and $n=0$;
\REPEAT
\IF {$\mathcal{P}_{3}$ is feasible}
\STATE Solve $\mathcal{P}_{3}$ for a given $q$ and obtain the optimal solution $\Psi_{\mathcal{P}_{3}}^{*}$ and the optimal value $\mathcal{W}_{q}^{*}$;
\ELSE
\STATE Break;
\ENDIF
\IF {$\mathcal{W}_{q}^{*} \geq - \Delta$}
\STATE {${\rm{Convergence}}=\ $\textbf{true} and \textbf{return} $\Psi_{\mathcal{P}_{1}}^{*}=\Psi_{\mathcal{P}_{3}}^{*}$};
\ELSE
\STATE Set $q=\frac{\lim\nolimits_{T\rightarrow \infty}\frac{1}{T}\mathbb{E}^{\Psi_{\mathcal{P}_{3}}^{*}}\left[\sum_{t=1}^{T}d\left(\mathbf{s}_{t},\mathbf{a}_{t}\right)\right]}{\lim\nolimits_{T\rightarrow \infty}\frac{1}{T}\mathbb{E}^{\Psi_{\mathcal{P}_{3}}^{*}}\left[\sum_{t=1}^{T}n\left(\mathbf{s}_{t},\mathbf{a}_{t}\right)\right]}$ and $n=n+1$;
\STATE {${\rm{Convergence}}=\ $\textbf{false}};
\ENDIF
\UNTIL{${\rm{Convergence}}=\ $\textbf{true} or $n=I_{\max}$}
\end{algorithmic}
\end{algorithm}

In order to solve $\mathcal{P}_{3}$, we first specify the feasible action set $\mathcal{A}_{\mathbf{s}_{t}}^{\rm{A}}$: When $i_{t}=0, E_{\rm{s}}+E_{\rm{d}}\leq b_{t}<E_{\rm{s}}+E_{\rm{d}}+E_{\rm{f}}$, we have $\mathcal{A}_{\mathbf{s}_{t}}^{\rm{A}}=\{\langle 0, 0\rangle, \langle 1, 0\rangle \}$, which corresponds to the scenario that the packet has not been decoded correctly, while the battery energy is only sufficient for sampling and decoding, i.e., the ACK cannot be fed back due to energy shortage; when $i_{t}=0,b_{t}\geq E_{\rm{s}}+E_{\rm{d}}+E_{\rm{f}}$, we have $\mathcal{A}_{\mathbf{s}_{t}}^{\rm{A}}=\{\langle 0, 0\rangle, \langle 1, 0\rangle, \langle 1,1\rangle \}$; when $i_{t}=1,b_{t}\geq E_{\rm{f}}$, we have $\mathcal{A}_{\mathbf{s}_{t}}^{\rm{A}}=\{\langle 0,0\rangle,\langle 0,1\rangle \}$ since the packet has been well received; otherwise, $\mathcal{A}_{\mathbf{s}_{t}}^{\rm{A}}=\{\langle 0,0\rangle\}$.

Denote the state transition probability as $\mathbb{P}\left(\mathbf{s}'|\mathbf{s},\mathbf{a}\right)$, which is the probability the system will be in state $\mathbf{s}'$ in the next time slot, given the current state is $\mathbf{s}$ and action $\mathbf{a}$ is taken. We derive the expressions of $\mathbb{P}\left(\mathbf{s}'|\mathbf{s},\mathbf{a}\right)$, $d\left(\mathbf{s},\mathbf{a}\right)$, $n\left(\mathbf{s},\mathbf{a}\right)$ and $s\left(\mathbf{s},\mathbf{a}\right)$ in the following lemma.

\begin{lma}
For the ARQ protocol with adaptive feedback, $n\left(\mathbf{s},\mathbf{a}\right)=\bm{1}_{k=0}$, $s\left(\mathbf{s},\mathbf{a}\right)=p_{c}\bm{1}_{i=0,a_{s}=1}$ and $d\left(\mathbf{s},\mathbf{a}\right)=\bm{1}_{i=0,k=K-1}\left(\overline{p}_{c}\bm{1}_{a_{s}=1}+\bm{1}_{a_{s}=0}\right)$. And $\mathbb{P}\left(\mathbf{s}'|\mathbf{s},\mathbf{a}\right)$ is given as
\begin{eqnarray}
\begin{small}
\begin{split}
&\mathbb{P}\left(\mathbf{s}'|\mathbf{s},\mathbf{a}=\langle 0,0\rangle\right)=\\
&\begin{cases}
\begin{split}
&\bm{1}_{k'=\left(k+1\right){\rm{mod}}K}\cdot\\&\sum\nolimits_{\forall e}\bm{1}_{b'=\phi\left(b,0,e\right)}\mathbb{P}_{\mathcal{E}}\left(e\right),\end{split} &\begin{split}&i=i',k<K-1\ {\rm{or}}\\&i'=0,k=K-1\end{split}\\
0, &else,
\end{cases}\\
&\mathbb{P}\left(\mathbf{s}'|\mathbf{s},\mathbf{a}=\langle 1,0\rangle\right)=\\
&\begin{cases}
\begin{split}
&\sum\nolimits_{\forall e}\big[\bm{1}_{i'=0,b'=\phi\left(b,E_{\rm{s}},e\right)}\overline{p}_{c}+\\
&\bm{1}_{i'=1,b'=\phi\left(b,E_{\rm{s}}+E_{\rm{d}},e\right)}p_{c}\big]\mathbb{P}_{\mathcal{E}}\left(e\right),
\end{split} &\begin{split}&i=0,b\geq E_{\rm{s}}+E_{\rm{d}},\\&k<K-1,k'=k+1\end{split}\\
\begin{split}
&\bm{1}_{i'=0}\sum\nolimits_{\forall e}\big[\bm{1}_{b'=\phi\left(b,E_{\rm{s}},e\right)}\overline{p}_{c}\\
&+\bm{1}_{b'=\phi\left(b,E_{\rm{s}}+E_{\rm{d}},e\right)}p_{c}\big]\mathbb{P}_{\mathcal{E}}\left(e\right),
\end{split} &\begin{split}&i=0,b\geq E_{\rm{s}}+E_{\rm{d}},\\&k=K-1,k'=0\end{split}\\
0, &else,
\end{cases}\\
&\mathbb{P}\left(\mathbf{s}'|\mathbf{s},\mathbf{a}=\langle 1,1\rangle\right)=\\
&\begin{cases}
\begin{split}
&\sum\nolimits_{\forall e}\big[\bm{1}_{k'=0,b'=\phi\left(b,E_{\rm{s}}+E_{\rm{d}}+E_{\rm{f}},e\right)}p_{c}+\\
&\bm{1}_{k'=\left(k+1\right){\rm{mod}}K,b'=\phi\left(b,E_{\rm{s}},e\right)}\overline{p}_{c}\big]\mathbb{P}_{\mathcal{E}}\left(e\right),
\end{split} &\begin{split}&i=i'=0,\\&b\geq E_{\rm{s}}+E_{\rm{d}}+E_{\rm{f}}\end{split}\\
0, &else,
\end{cases}\\
&\mathbb{P}\left(\mathbf{s}'|\mathbf{s},\mathbf{a}=\langle 0,1\rangle\right)=\\
&\begin{cases}
\sum\nolimits_{\forall e}\bm{1}_{k'=0,b'=\phi\left(b,E_{\rm{f}},e\right)}\mathbb{P}_{\mathcal{E}}\left(e\right), &\begin{split}&i=1,i'=0,\\&k>0,b\geq E_{\rm{f}},\end{split}\\
0, &else,\\
\end{cases}
\end{split}\nonumber
\end{small}
\end{eqnarray}
where $\phi\left(b,u,e\right)\triangleq \min\{b-u+e,B_{\max}\}$.
\end{lma}
\begin{proof}
By definition, $n\left(\mathbf{s},\mathbf{a}\right)=\bm{1}_{k=0}$ holds. As discussed in Section IV-A, when $a_{s}=1$, the packet will be decoded correctly with probability $p_{c}$, while $a_{s}=0$, no packet will be decoded, which lead to the expressions of $s\left(\mathbf{s},\mathbf{a}\right)$ and $d\left(\mathbf{s},\mathbf{a}\right)$. The state transition when $\mathbf{a}_{t}=\langle 0,0\rangle$ and $\langle 1,0\rangle$ are similar to those in ARQ without feedback. When $\mathbf{a}_{t}=\langle 1,1\rangle$, if $|h_{t}|^{2}\geq |h_{\rm{th}}|^{2}$, the packet can be decoded correctly and $E_{\rm{s}}+E_{\rm{d}}+E_{\rm{f}}$ units of energy will be consumed which accounts for sampling, decoding and feedback, i.e., both $k$ and $i$ will be $0$ in the next time slot. When $\mathbf{a}_{t}=\langle 0,1\rangle$, $E_{\rm{f}}$ units of energy will be consumed and the transmitter will transmit a new packet in the next time slot. Detailed derivation for $\mathbb{P}\left(\mathbf{s}'|\mathbf{s},\mathbf{a}\right)$ is omitted due to space limitation.
\end{proof}

It is not difficult to identify that $\mathcal{P}_{3}$ is a constrained MDP (CMDP) problem. In the following, we provide a linear programming (LP) approach to obtain the optimal solution for $\mathcal{P}_{3}$ by solving the following LP problem $\mathcal{P}_{4}$ \cite{EAltman99}:

\vspace{-10pt}
\begin{small}
\begin{align}
&\mathcal{P}_{4}:\min_{x\left(\mathbf{s},\mathbf{a}\right)}\sum_{\forall \mathbf{s}}\sum_{\mathbf{a}\in\mathcal{A}_{\mathbf{s}}^{\rm{A}}}c\left(\mathbf{s},\mathbf{a}\right){x\left(\mathbf{s},\mathbf{a}\right)}\label{averageC}\\
&\ \ \ \ \ \ {\rm{s. t.}}\ \sum_{\mathbf{a}\in\mathcal{A}_{\mathbf{s}'}^{\rm{A}}}{x\left(\mathbf{s}',\mathbf{a}\right)}=\sum_{\forall \mathbf{s}}\sum_{\mathbf{a}\in\mathcal{A}_{\mathbf{s}}^{\rm{A}}}\mathbb{P}\left(\mathbf{s}'|\mathbf{s},\mathbf{a}\right){x\left(\mathbf{s},\mathbf{a}\right)},\forall \mathbf{s}'\label{balanceequation}\\
&\ \ \ \ \ \ \ \ \ \ \ \sum_{\forall \mathbf{s}}\sum_{\mathbf{a}\in\mathcal{A}_{\mathbf{s}}^{\rm{A}}}s\left(\mathbf{s},\mathbf{a}\right)x\left(\mathbf{s},\mathbf{a}\right)\geq \mathcal{T}_{\rm{th}}\label{xtpconstraint}\\
&\ \ \ \ \ \ \ \ \ \ \ \sum_{\forall \mathbf{s}}\sum_{\mathbf{a}\in\mathcal{A}_{\mathbf{s}}^{\rm{A}}}x\left(\mathbf{s},\mathbf{a}\right)=1,0\leq x\left(\mathbf{s},\mathbf{a}\right)\leq \bm{1}_{\mathbf{a}\in\mathcal{A}_{\mathbf{s}}^{\rm{A}}},\forall \mathbf{s},\mathbf{a}\label{definex}
\end{align}
\end{small}where $c\left(\mathbf{s},\mathbf{a}\right)\triangleq d\left(\mathbf{s},\mathbf{a}\right)-q\cdot n\left(\mathbf{s},\mathbf{a}\right)$, and $x\left(\mathbf{s},\mathbf{a}\right)$ is the occupation measure that gives the steady state probability that the system is in state $\mathbf{s}$ and action $\mathbf{a}$ is chosen. (\ref{averageC}) is an equivalent representation of $\lim\limits_{T\rightarrow \infty}\frac{1}{T} \mathbb{E}\left[\sum_{t=1}^{T}c\left(\mathbf{s}_{t},\mathbf{a}_{t}\right)\right]$, while (\ref{balanceequation}) is the balance equation. With the optimal solution to $\mathcal{P}_{4}$, $x^{*}\left(\mathbf{s},\mathbf{a}\right)$, we can construct an optimal stationary and randomized policy to $\mathcal{P}_{3}$ as $\Psi_{\mathcal{P}_{3}}^{*}:\mathbf{s}\rightarrow \mathbf{a}$ with probability $x^{*}{\left(\mathbf{s},\mathbf{a}\right)}\slash\sum_{\mathbf{a}\in\mathcal{A}_{\mathbf{s}}^{\rm{A}}}x^{*}\left(\mathbf{s},\mathbf{a}\right)$ (If $\sum_{\mathbf{a}\in\mathcal{A}_{\mathbf{s}}^{\rm{A}}}x^{*}\left(\mathbf{s},\mathbf{a}\right)=0$, $\mathbf{a}$ can be an arbitrary feasible action and set to be $\langle 0,0\rangle$). Besides, the update of $q$ in line 11 of Algorithm 1 can be rewritten as $q=\sum_{\mathbf{s},\mathbf{a}\in\mathcal{A}_{\mathbf{s}}^{\rm{A}}}d\left(\mathbf{s},\mathbf{a}\right)x^{*}\left(\mathbf{s},\mathbf{a}\right)\slash\sum_{\mathbf{s},\mathbf{a}\in\mathcal{A}_{\mathbf{s}}^{\rm{A}}}n\left(\mathbf{s},\mathbf{a}\right)x^{*}\left(\mathbf{s},\mathbf{a}\right)
$.

\begin{rmk}
For systems with an energy-unconstrained receiver, the myopic policy is optimal, i.e., $\mathbf{a}_{t}=\langle 1,0\rangle$ and $\langle 1,1\rangle,\forall t$ for ARQs without feedback and with non-adaptive feedback, respectively. Based on the renewal theory, the PDPs achieved by ARQs without feedback and with non-adaptive feedback are the same as $\overline{p}_{c}^{K}$, while the throughput achieved by ARQ with non-adaptive feedback is $p_{c}$, which is greater than $\left(1-\overline{p}_{c}^{K}\right)\slash K$ achieved by the one without feedback.
\end{rmk}

\section{Simulation Results}
In simulations, we assume the channel is Rayleigh fading with unit variance, $R=0.5$ bps\slash Hz, and $p_{\rm{tx}}=1$ W. The EH process is modeled as an i.i.d. Bernoulli process with EH probability $\rho$ (also termed as the normalized EH rate) \cite{YMao1412}, i.e., $\mathcal{E}\!=\!\{0,\hat{e}\}$, $\mathbb{P}_{\mathcal{E}}\left(0\right)\!=\!1-\rho$ and $\mathbb{P}_{\mathcal{E}}\left(\hat{e}\right)=\rho$. We use $\hat{e}=6E$ as an example, and set $B_{\max}=15E$, $E_{\rm{d}}=E_{\rm{s}}=3E$, $I_{\max}=20$, and $\Delta=10^{-6}$. For comparison, we introduce the myopic policies for the ARQ protocols with feedback: For ARQ with non-adaptive feedback, $\mathbf{a}_{t}\!=\!\langle 1,1\rangle$ when $i_{t}=0,b_{t}\geq E_{\rm{s}}\!+\!E_{\rm{d}}\!+\!E_{\rm{f}}$; while for ARQ with adaptive feedback, $\mathbf{a}_{t}\!=\!\langle 1,0\rangle$ when $i_{t}\!=\!0$, $E_{\rm{s}}+E_{\rm{d}}\!\leq\! b_{t}\!<\! E_{\rm{s}}\!+\!E_{\rm{d}}\!+\!E_{\rm{f}}$, $\mathbf{a}_{t}=\langle 1,1\rangle$ when $i_{t}\!=\!0$, $b_{t}\!\geq\! E_{\rm{s}}\!+\!E_{\rm{d}}\!+\!E_{\rm{f}}$, and $\mathbf{a}_{t}=\langle 0,1\rangle$ when $i_{t}=1,b_{t}\geq E_{\rm{f}}$.


\begin{figure}[h]
\centering
\includegraphics[width=0.45\textwidth]{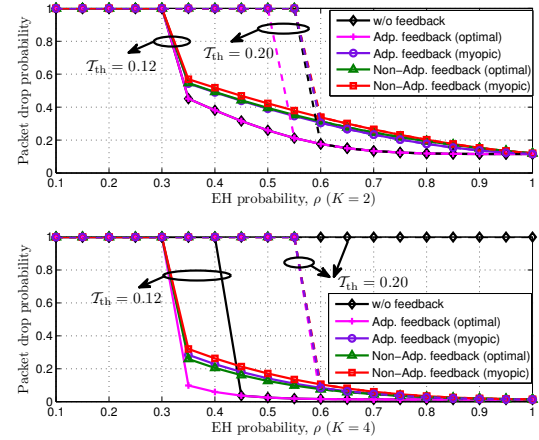}
\label{outage_EHrateTP}
\vspace{-10pt}
\caption{PDP vs. EH probability $\rho$, $E_{\rm{f}}=E$. We set the PDP to 1 if the throughput constraint can not be satisfied.}
\end{figure}

The PDPs achieved by different ARQs are shown in Fig. 1. In general, a higher EH rate is needed in order to meet a higher throughput requirement. Also, with a larger $K$, i.e., the maximum number of transmission attempts, the system has the potential to achieve a smaller PDP, but the throughput requirement becomes more difficult to satisfy. Besides, when the throughput constraint is met, ARQ without feedback and ARQ with adaptive feedback achieve the same performance under the optimal policies. This is because each packet uses all the transmission attempts to avoid packet drop, and the proposed ARQ protocol reduces to the ARQ without feedback. From this set of results, we can draw new design insights for communication systems with EH receivers:
\begin{itemize}
\item In contrast to systems with energy-unconstrained receivers, for which the myopic policies are optimal, with EH receivers, noticeable performance gains are provided by the optimal reception policies for the ARQs with feedback compared to the myopic policies. This shows the significant difference in retransmission policy design with EH receivers.

\item When the required throughput is achieved, the ARQ without feedback outperforms the one with non-adaptive feedback, which is different from systems with energy-unconstrained receivers, where both protocols achieve the same PDP as discussed in Remark 2. This is due to the \emph{feedback energy consumption} and the \emph{waste of transmission attempts} incurred by energy shortage at the receiver. Thus ARQ without feedback may be beneficial in certain scenarios with EH receivers.
\item Compared to the ARQ protocol with non-adaptive feedback, significant performance improvement is achieved by the proposed ARQ protocol, and a smaller EH rate is needed to meet the throughput requirement. These confirm the importance of intelligent feedback strategies for EH receivers and reveal the difference from systems with energy-unconstrained receivers, where the proposed ARQ protocol reduces to the one with non-adaptive feedback as emphasized in Remark 1.
\end{itemize}

\begin{figure}[h]
\centering
\includegraphics[width=0.45\textwidth]{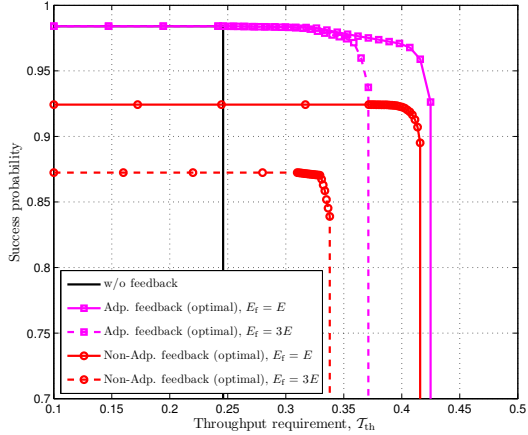}
\label{suc_Tp}
\vspace{-10pt}
\caption{Success probability vs. $\mathcal{T}_{\rm{th}}$, $K=4$ and $\rho=0.6$.}
\end{figure}

In Fig. 2, we show the achievable success probability\footnote{Success probability $\triangleq 1-$ PDP.}-throughput (S-T) regions for different ARQ protocols. Key observations can be drawn:
\begin{itemize}
\item The S-T region of the ARQ without feedback is a rectangle since any PDP-optimal policy is also throughput-optimal as discussed in Section IV-A.
\item For ARQs with feedback, a higher value of $E_{\rm{f}}$ leads to a smaller S-T region, and there exists a tradeoff between the achievable success probability and throughput, i.e., a proper operating point should be chosen to balance the system's reliability and efficiency.
\item With a small value of $\mathcal{T}_{\rm{th}}$, ARQ with adaptive feedback reduces to the one without feedback, similar as the case in  Fig. 1; while with a large value of $\mathcal{T}_{\rm{th}}$, the proposed ARQ protocol not only enjoys a higher success probability, but it is also competent to meet the throughput requirements compared to ARQ with non-adaptive feedback.
\end{itemize}

\section{Conclusions}
In this paper, we proposed a novel ARQ protocol with adaptive feedback for EH receivers, which offers the option of ACK feedback management and helps utilize the harvested energy more efficiently. The throughput constrained packet drop probability minimization problem was investigated, and the optimal reception policies were developed. Simulation results demonstrated the benefits of adaptive feedback in minimizing the wastage of the harvested energy and transmission attempts. In addition, this investigation has revealed new design insights for communication systems with EH receivers, and indicated the importance of a comprehensive consideration of different components of energy consumption.

\begin{appendix}
\textbf{Proof for Proposition 1:}
The optimal $a_{s,t}$ should maximize the long-term average expected throughput, i.e., $a_{s,t}^{*}=\arg\max\nolimits_{a\in\mathcal{A}_{\mathbf{s}_{t}}^{\rm{wo}}}\lim\nolimits_{T\rightarrow \infty}\frac{V_{T}^{a}}{T-t+1}$, where
$V_{T}^{a}\triangleq s\left(\mathbf{s}_{t},a\right)+\mathbb{E}\left[\sum_{\tau=t+1}^{T}s\left(\mathbf{s}_{\tau},a_{s,\tau}^{*}\right)|\mathbf{s}_{t},a\right]$. The first term in $V_{T}^{a}$ denotes the expected number of successfully decoded packets in time slot $t$, while the second term stands for the expected number of total successfully decoded packets from the $t+1$ to the $T$th time slot under the optimal policy $\{a_{s,\tau}^{*}\}$. In the following, we will show if $a_{s,t}=1\in\mathcal{A}_{\mathbf{s}_{t}}^{\rm{wo}}$, then $V_{T}^{1}\geq V_{T}^{0},\forall T>t$.

Suppose $a_{s,t}=0$ and let $t_{1}=\min\{\tau \geq t+1|\hat{a}_{s,\tau}^{*}=1\}$, where $\hat{a}_{s,\tau}^{*},t+1 \leq \tau \leq T$, is a sample-path of the optimal sampling policy given $a_{s,t}=0$. For convenience, $t_{1}\triangleq T+1$ if $\hat{a}_{s,\tau}^{*}=0, \forall \tau \leq T$. Note that $t_{1}$ is random and takes value from $\{t+1,\cdots,T+1\}$. Thus, $V_{T}^{0}=\mathbb{E}_{t_{1}}\left[s\left(\mathbf{s}_{t_{1}},1\right)\cdot \bm{1}_{t_{1}\leq T} + V_{T,t_{1}}^{0}\right]$,
where $V_{T,t_{1}}^{0}=\mathbb{E}\left[\sum_{\tau=t_{1}+1}^{T}s\left(\mathbf{s}_{\tau},\hat{a}_{s,\tau}^{*}\right)|\mathbf{s}_{t},\hat{a}_{s,\tilde{\tau}}^{*},\tilde{\tau}\leq t_{1}\right]\cdot \bm{1}_{t_{1}<T}$.

Suppose $a_{s,t}=1$, then $a_{s,\tau}=0,\forall t<\tau\leq  t_{1}$ is feasible for a given $t_{1}$. Therefore, $V_{T}^{1}\geq  s\left(\mathbf{s}_{t},1\right)+ \mathbb{E}_{t_{1}}\left[V_{T,t_{1}}^{1}\right]$,
where $V_{T,t_{1}}^{1}=\mathbb{E}\left[\sum_{\tau=t_{1}+1}^{T}s\left(\mathbf{s}_{\tau},\tilde{a}_{s,\tau}^{*}\right)|\mathbf{s}_{t},a_{s,\tilde{\tau}},\tilde{\tau}\leq t_{1}\right]\cdot \bm{1}_{t_{1}<T}$,
and $\tilde{a}^{*}_{s,\tau}, t_{1}+1 \leq \tau\leq T$, denotes the optimal sampling policy given $a_{s,t}=1$ and $a_{s,\tau}=0,t<\tau\leq t_{1}$.

By expanding $V_{T,t_{1}}^{0}$ and $V_{T,t_{1}}^{1}$ conditioned on the events $S_{t_{1}}$ and $S_{t}$, respectively, and utilizing the fact that $B_{\max}$ is finite, we have $s\left(\mathbf{s}_{t},1\right)+ V_{T,t_{1}}^{1} \geq  s\left(\mathbf{s}_{t_{1}},1\right)\cdot \bm{1}_{t_{1}\leq T} + V_{T,t_{1}}^{0}, \forall t_{1}$.
By taking the expectation with respect to $t_{1}$ for both sides of the inequality, we can conclude $V_{T}^{1}\geq V_{T}^{0}, \forall T>t$, i.e., $a_{s,t}=1$ performs no worst than $a_{s,t}=0$, which ends the proof.
\end{appendix}

\end{document}